\documentclass[a4paper,11pt]{article}
\pdfoutput=1
\usepackage{graphicx}
\usepackage{amsmath,amssymb,amsthm}
\usepackage[numbers]{natbib}
\usepackage[english]{babel}
\usepackage{subfigure}
\usepackage{hyperref}

\newcommand{\vc}{{\cal V}_C}
\newcommand{\vr}{{\cal V}_R}
\newtheorem{thm}{Theorem}
\newtheorem{lem}[thm]{Lemma}

\title{An Approximation Ratio for Biclustering\footnote{Cite as 
\href{http://arxiv.org/abs/0712.2682v2}{arXiv:0712.2682v2 [cs.DS]}. 
This is a revised version of the preprint originally published as
\href{http://arxiv.org/abs/0712.2682v1}{arXiv:0712.2682v1 [cs.DS]}
on 17 December 2007. 
To be published in Information Processing Letters 108 (2008) 45--49
[\href{http://dx.doi.org/10.1016/j.ipl.2008.03.013}{doi:10.1016/j.ipl.2008.03.013}].}}

\author{Kai Puolam\"aki\footnote{Kai.Puolamaki@tkk.fi}\\
Sami Hanhij\"arvi\footnote{Sami.Hanhijarvi@tkk.fi}\\
Gemma C.~Garriga\footnote{Gemma.Garriga@tkk.fi}\\
~\\
Helsinki Institute for Information Technology HIIT\\
Helsinki University of Technology\\
P.O. Box 5400\\
FI-02015 TKK\\
Finland}

\date{22 August 2008}

\begin{document}

\maketitle

\begin{abstract}

The problem of biclustering consists of the simultaneous clustering of rows and
columns of a matrix such that each of the submatrices induced
by a pair of row and column clusters is as uniform as possible.
In this paper we approximate the optimal biclustering by applying one-way
clustering algorithms independently on the rows and on the columns
of the input matrix. 
We show that such a solution yields a worst-case approximation
ratio of $1+\sqrt 2$ under $L_1$-norm for 0--1 valued matrices, and
of $2$ under $L_2$-norm for real valued matrices.

\end{abstract}

~


\section{Introduction}

The standard clustering problem \cite{JainDubes88} consists of
partitioning a set of input vectors, such that the vectors in each
partition (cluster) are close to one another according to some
predefined distance function. This formulation is the objective of the
popular $K$-means algorithm (see, for example, \cite{Kanungo04}), where
$K$ denotes the final number of clusters and the distance function is
defined by the $L_2$-norm. Another similar example of this formulation
is the $K$-median algorithm (see, for example, \cite{Arya01}), where the
distance function is given by the $L_1$-norm. Clustering a set of
input vectors is a well-known NP-hard problem even for $K=2$ clusters
\cite{Drineas04}. Several approximation guarantees have been shown for
this formulation of the standard clustering problem (see \cite{Arya01,
  Kanungo04,Arthur07} and references therein).

Intensive recent research
has focused on the discovery of homogeneous 
substructures in large matrices. This is also one of the goals in the 
problem of {\em biclustering}.
Given a set of $N$ rows in $M$ columns from a matrix $X$,
a biclustering algorithm identifies subsets of 
rows exhibiting similar behavior across a subset of columns, or vice versa. 
Note that the optimal solution for this problem 
necessarily requires to cluster the $N$ vectors and the $M$ dimensions
simultaneously, thus the name biclustering.  
Each submatrix of $X$, induced by a pair of row and column clusters,
is typically referred to as a {\em bicluster}. 
See Figure \ref{fig:ex} for a simple toy example. 
The main challenge of a biclustering algorithm lies in
the dependency between the row and column partitions, which makes
it difficult to identify the optimal biclusters. A
change in a row clustering affects the cost of the induced
submatrices (biclusters), and as a consequence, the column clustering
may also need to be changed to improve the solution.

\begin{figure*}
	\begin{center}
		\subfigure[]{\includegraphics[width=0.25\textwidth]{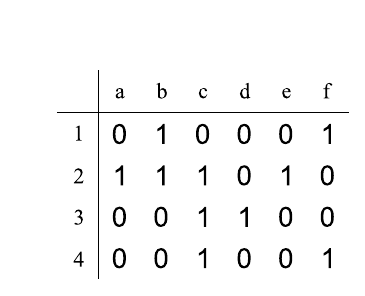}}\hspace{0.04\textwidth}\subfigure[]{\includegraphics[width=0.25\textwidth]{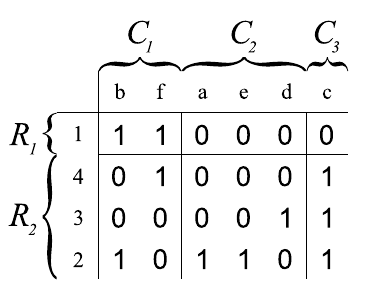}}\hspace{0.04\textwidth}\subfigure[]{\includegraphics[width=0.25\textwidth]{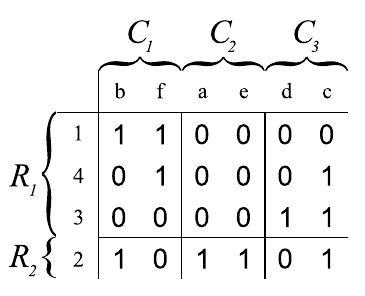}}
		\caption{\label{fig:ex} (a) An example binary data matrix $X$ of dimensions $4\times6$, with rows and columns
labeled with numbers and characters. (b) The optimal biclustering of $X$ consists of 
$\{R^*_1,R^*_2\}=\{\{1\},\{2,3,4\}\}$ row clusters 
		and $\{C^*_1,C^*_2,C^*_3\} = \{\{b,f\},\{a,d,e\},\{c\}\}$ column clusters when using $L_1$-norm.
		(c) Biclusters of the data matrix returned by our scheme, that is, using twice
		an optimal one-way clustering algorithm, once on the $4$ row vectors and another on the $6$ column vectors, 
		with $L_1$-norm. Resulting clusterings are $\{R_1,R_2\}=\{\{1,3,4\},\{2\}\}$ for rows and 
		$\{C_1,C_2,C_3\}=\{\{b,f\},\{a,e\},\{d,e\}\}$ for columns.
		For visual clarity, the rows and columns of the original 
		matrix in (a) have been permuted in (b) and (c) by making the rows (and columns) of a single cluster adjacent.
			}	
	\end{center}
\end{figure*}

Finding an optimal solution for the biclustering problem is
NP-hard. This observation follows directly from the reduction of the 
standard clustering problem (known to be NP-hard) to the
biclustering problem  by fixing the number of clusters in
columns to $M$. 
To the best of our knowledge, no algorithm exists that can efficiently
approximate biclustering with a proven approximation ratio. The goal
of this paper is to propose
such an approximation guarantee by means of a very simple scheme.

Our approach will consist of relieving the requirement for simultaneous 
clustering of rows and columns and instead perform them independently. 
In other words, our final biclusters will correspond to the 
submatrices of $X$ induced by pairs of row
and columns clusters, found independently with a standard clustering algorithm. 
We sometimes refer to this standard clustering algorithm as one-way clustering. 
The simplicity of the solution alleviates us
from the inconvenient dependency of rows and columns. 
More importantly, the solution obtained with this approach, despite not 
being optimal, allows for the study of
approximation guarantees on the obtained biclusters. 
Here we prove that our solution achieves a worst-case approximation
ratio of $1+\sqrt 2$ under $L_1$-norm for 0--1 valued matrices, and
of $2$ under $L_2$-norm for real valued matrices.

Finally, note that our final solution is constructed on top of a
standard clustering algorithm (applied twice, once in row vectors and the other
in column vectors) and therefore, it is necessary to multiply our
ratio with the approximation ratio achieved by the used standard
clustering algorithm (such as
\cite{Arya01,Kanungo04}). 
For clarity, 
we will lift this restriction in the following proofs 
by assuming that the applied one-way clustering algorithm 
provides directly an optimal solution
to the standard clustering problem.

\subsection{Related work}

This basic algorithmic problem and several variations were initially
presented in \cite{Hartigan72} with the name of direct clustering. The
same problem and its variations have also been referred to as two-way
clustering, co-clustering or subspace clustering. In practice,
finding highly homogeneous biclusters has important applications in
biological data analysis (see \cite{Madeira2004} for review and
references), where a bicluster may, for example, correspond to an activation
pattern common to a group of genes only under specific experimental
conditions.

An alternative definition of the basic biclustering problem described in
the introduction consists
on finding the maximal bicluster in a given matrix. A well-known connection
of this alternative formulation is its reduction
to the problem of finding a biclique in a bipartite graph~\cite{Hochbaum98}. 
Algorithms for detecting
bicliques enumerate them in the graph by
using the monotonicity property that a subset of a biclique is also a
biclique~\cite{Alexe04,Eppstein94}. These algorithms usually have a high order of complexity.

\section{Definitions}

We assume given a matrix $X$ of size $N\times M$, and integers $K_r$ and $K_c$, 
which define the number of clusters partitioning rows and columns, respectively. 
The goal is to approximate the optimal biclustering of $X$ by
means of a one-way row clustering into $K_r$ clusters and a
one-way column clustering into $K_c$ clusters.

For any $T \in \mathbb{N}$ we denote $[T] = \{1,\ldots,T\}$. We use 
$X(R,C)$, where $R\subseteq [N]$ and $C\subseteq [M]$, to
denote the submatrix of $X$ induced by the subset of rows $R$ and the
subset of columns $C$. Let $Y$ denote an induced submatrix of $X$,
that is $Y = X(R,C)$ for some
$R\subseteq [N]$ and $C\subseteq [M]$.
When required by the context, we will also refer to $Y = X(R,C)$ as a
bicluster of $X$ 
and denote the size of $Y$ with $n \times m$, where $n \leq N$ and $m \leq N$.
We use ${\rm{median}}(Y)$ and ${\rm{mean}}(Y)$ to denote
the median and mean of all elements of $Y$, respectively.

The scheme for approximating the optimal biclustering is defined as follows.\vspace{5mm}\\
\begin{tabular}{p{0.94\linewidth}} 
\hline 
\textbf{Input:} matrix $X$, number of row clusters $K_r$, number of column clusters $K_c$\\ \\
${\cal R}=\mathrm{kcluster}(X,K_r)$\\
${\cal C}=\mathrm{kcluster}(X^T,K_c)$\\ \\
\textbf{Output:} a set of biclusters $X(R,C)$, for each $R \in {\cal R}$, $C \in {\cal C}$
\\
\hline
\end{tabular}
\vspace{5mm}

The function $\mathrm{kcluster}(X,K_r)$ denotes here 
an optimal one-way clustering algorithm 
that partitions the row vectors of matrix 
$X$ into $K_r$ clusters. 
We have used $X^T$ to denote the transpose of matrix $X$. 

Instead of fixing a specific norm for the formulas, we use the 
dissimilarity measure ${\cal V}()$ to absorb the norm-dependent part. 
For $L_1$-norm, ${\cal V}()$ would be defined as 
${\cal V}(Y)= \sum_{y\in Y}{\left|y-{\rm{median}}(Y)\right|}$,
and for $L_2$-norm as 
${\cal V}(Y)= \sum_{y\in Y}{\left(y-{\rm{mean}}(Y)\right)^2}$.
Given $Y$ of size $n\times m$, we further use a special
row norm, $\vr (Y)= \sum_{j=1}^m{{\cal V}(Y([n],j))}$, 
and a special column norm, $\vc (Y)=\sum_{i=1}^n{{\cal V}(Y(i,[m]))}$.

We define the one-way row clustering, given by kcluster above,
as a partition of rows $[N]$ into
$K_r$ clusters ${\cal R}=\{R_1,\ldots,R_{K_r}\}$ 
such that the cost function
\begin{equation}
L_R = \sum_{R\in{\cal R}}{\sum_{j=1}^M{{\cal V}(X(R,j))}}
\label{eq:lr}
\end{equation}
is minimized. Analogously, the one-way clustering of columns $[M]$ 
into $K_c$ clusters ${\cal C}=\{C_1,\ldots,C_{K_c}\}$
is defined such that the cost function
\begin{equation}
L_C = \sum_{i=1}^{N}{\sum_{C\in{\cal C}}{{\cal V}(X(i,C))}}
\label{eq:lc}
\end{equation}
is minimized.

The cost of biclustering, induced by the two one-way clusterings
above, is
\begin{equation}
L=\sum_{R\in{\cal R}}{\sum_{C\in{\cal C}}{{\cal V}(X(R,C))}}.
\label{eq:l}
\end{equation}

Notice that we are assuming that the one-way clusterings above,
denoted ${\cal R}$ on rows and ${\cal C}$ on columns, 
correspond to optimal one-way partitionings on rows and columns,
respectively. 	

Finally, the optimal biclustering on $X$ is given
by simultaneous row and column partitions ${\cal R}^*=\{R_1^*,\ldots,R_{K_r}^*\}$ and
${\cal C}^*=\{C_1^*,\ldots,C_{K_c}^*\}$, that minimize the cost
\begin{equation}
L^{*}=\sum_{R^*\in{\cal R}^*}{\sum_{C^*\in{\cal C}^*}{{\cal V}(X(R^*,C^*))}}.
\end{equation}

\section{Approximation ratio}
\label{sec:lg}

Given the definitions above, our main result reads as follows.
\begin{thm}
  There exists an approximation ratio of $\alpha$ such that $L\le\alpha
  L^{*}$, where $\alpha=1+\sqrt 2\approx 2.41$ for $L_1$-norm and
  $X\in \{0,1\}^{N\times M}$, and $\alpha=2$ for $L_2$-norm and
  $X\in\mathbb{R}^{N\times M}$.
\label{thm:main}
\end{thm}
We use the following intermediate result to prove the theorem.
\begin{lem}
There exists an approximation ratio of at most $\alpha$, that is,
$L\le\alpha L^*$,
if for any $X$ and for any partitionings ${\cal R}$ and ${\cal C}$ of $X$, all
biclusters $Y=X(R,C)$, with $R\in{\cal R}$ and $C\in{\cal C}$, satisfy
\begin{equation}
{\cal V}(Y)\le\frac 12 \alpha\left(\vr (Y)+\vc(Y) \right).
\label{eq:thbound}
\end{equation}
\label{lem:alpha}
\end{lem}

\begin{proof}
First we note that the cost of the optimal biclustering $L^*$ cannot
increase when we increase the number of row (or column) clusters.  For
example, consider the special case where $K_r=N$ (or $K_c=M$). In such
case, each row (or column) is assigned to its own cluster and the cost
of the optimal biclustering equals the cost of the optimal
one-way clustering on columns $L_C$ (or rows $L_R$). Hence,
the optimal biclustering solution is bounded from below by
\begin{equation}
L^*\ge \max{\left(L_R,L_C\right)}\ge \frac 12\left(L_R+L_C\right)
\label{eq:lower}
\end{equation}

Summing both sides of Equation (\ref{eq:thbound}),
\begin{eqnarray*}
\sum_{R\in {\cal R}}{\sum_{C\in{\cal C}}{{\cal V}(Y)\arrowvert_{Y=X(R,C)}}} \le
\frac 12 \alpha \sum_{R\in {\cal R}}{\sum_{C\in{\cal C}}{\left(\vr (Y)+\vc (Y)\right)\arrowvert_{Y=X(R,C)}}},
\end{eqnarray*}
and using Equations (\ref{eq:lr}), (\ref{eq:lc}) and
(\ref{eq:l}), gives $L\le \frac 12\alpha\left(L_R+L_C\right)$, which
together with Equation (\ref{eq:lower}) implies the approximation ratio
of $L\le\alpha L ^*$.
\end{proof}

Theorem \ref{thm:main} is proven separately in Sections \ref{sec:l1} and \ref{sec:l2}
using Lemma \ref{lem:alpha}. Section \ref{sec:l1} deals with the case of having
a 0--1 valued matrix $X$ and $L_1$-norm distance function, while Section \ref{sec:l2} deals with 
real valued matrix $X$ and $L_2$-norm.

\subsection{$L_1$-norm and 0--1 valued matrix}
\label{sec:l1}

\begin{figure*}[ht!]
\begin{center}
	\subfigure[] {
		\includegraphics[width=0.2\linewidth]{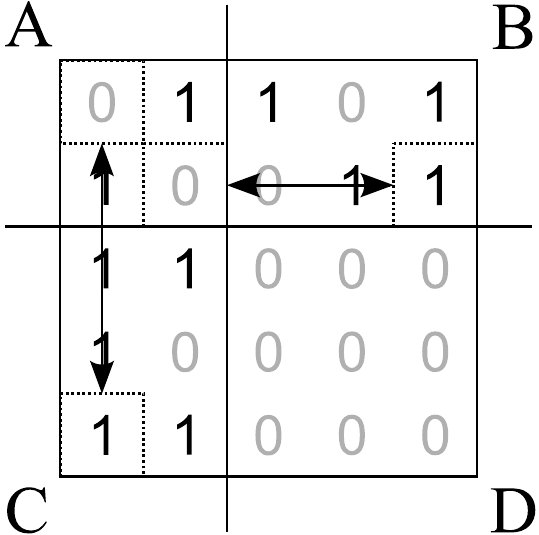}
	} \hspace{0.225\linewidth} 
	\subfigure[] {
		\includegraphics[width=0.2\linewidth]{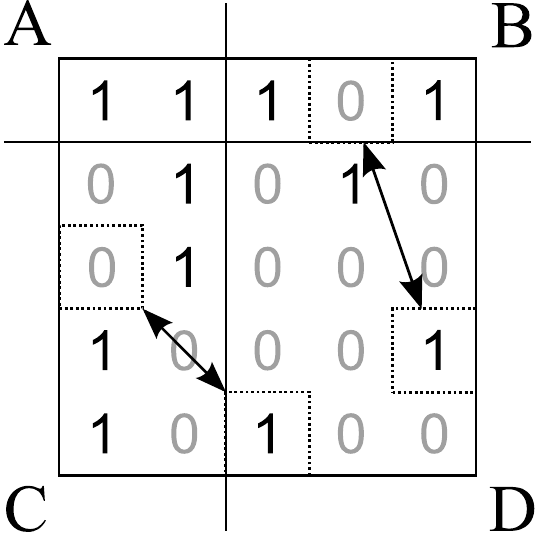}
	}
	\caption{Examples of swaps performed within bicluster $Y$ for the technical part of the proof in
          Section \ref{sec:l1}. For clarity, the
          rows and columns of the bicluster $Y$ have been ordered such
          that the blocks $A$, $B$, $C$ and $D$ are continuous.}
	\label{fig:swaps}
\end{center}
\end{figure*}

Consider a 0--1 valued matrix $X$ and $L_1$-norm. To prove Theorem \ref{thm:main} 
it suffices to show that Equation
(\ref{eq:thbound}) holds for each of the biclusters $Y=X(R,C)$ of $X$,
where $R\in{\cal R}$ and $C\in{\cal C}$. Therefore, in the following we concentrate 
on one single bicluster $Y\in{\{0,1\}}^{n\times m}$.

Without loss of generality, we consider only the case where the bicluster
$Y$ has at least as many 0's as 1's. In such case, the median of $Y$
can be safely taken to be zero and the cost ${\cal V}(Y)\le\frac 12nm$ is then
fixed to the number of 1's in the matrix. 
To get the worst
case scenario towards the tightest upper bound on $\alpha$ in Equation
(\ref{eq:thbound}), we should find first a configuration of 1's such
that, given ${\cal V}(Y)$, the sum $\vr (Y)+\vc (Y)$ is
minimized. 

Denote by $O_R$ and $O_C$ the sets of rows and columns in $Y$ which
have more 1's than 0's, respectively. Denote $A=Y(O_R,O_C)$,
$B=Y(O_R,[m]\setminus O_C)$, $C=Y([n]\setminus O_R,O_C)$,
$D=Y([n]\setminus O_R,[m]\setminus O_C)$, $n'=|O_R|$ and $m'=|O_C|$.
Note that $A$, $B$, $C$ and $D$ are simply blocks of bicluster $Y$,
which we need to make explicit in our notation for the proof.

Changing a 0 to 1 in $A$ or a 1 to 0 in $D$ decreases $\vr (Y)+\vc
(Y)$ by two, while changing a 0 to 1 or 1 to 0 in $B$ or $C$ changes
$\vr (Y)+\vc (Y)$ by at most one.  It follows that {\em swapping} a 1
in $B$ or $C$ with a 0 in $A$ (see Figure \ref{fig:swaps}a), or
swapping a 1 in $D$ with a 0 in $A$, $B$ or $C$ (see Figure
\ref{fig:swaps}b) decreases $\vr (Y)+\vc (Y)$ while ${\cal V}(Y)$
remains unchanged.  In other words, in a solution that minimizes $\vr
(Y)+\vc (Y)$ no such swaps can be made. In the remainder of this
subsection, we assume that the bicluster $Y$ satisfies this mentioned
property.

It follows that (i) $A$, $B$ and $C$ are blocks of 1's, (ii) $A$ is
a block of 1's and $D$ is a block of 0's, or (iii) $B$, $C$ and $D$
are blocks of 0's. Denote by $o()$ the number of 1's in a given
block. It follows that ${\cal V}(Y)=o(A)+o(B)+o(C)+o(D)\le\frac
12nm$, $\vr (Y)=nm'-o(A)+o(B)-o(C)+o(D)$ and $\vc
(Y)=n'm-o(A)-o(B)+o(C)+o(D)$. We denote $x=n'/n$, $y=m'/m$, $a=o(A)/(nm)$,
$b=o(B)/(nm)$, $c=o(C)/(nm)$ and $d=o(D)/(nm)$ and rewrite Equation
(\ref{eq:thbound}) as
\begin{eqnarray*}
\alpha &=& \sup{\left(\frac{2{\cal V}(Y)}{\vr (Y)+\vc (Y)}\right)} \\
  &=&2\sup{\left(\frac{a+b+c+d}{x+y-2a+2d}\right)},
\end{eqnarray*}
with constraints $a+b+c+d\in[0,\frac 12]$, $x\in[0,1]$ $y\in[0,1]$, as
well as (i) $a=xy$, $b=x(1-y)$, $c=(1-x)y$ and $d\in[0,(1-x)(1-y)]$;
(ii) $a=xy$, $b\in[0,x(1-y)]$, $c\in[0,(1-x)y]$ and $d=0$; or (iii)
$a\in[0,xy]$ and $b=c=d=0$.  The optimization problem has two
solutions, (i) $x=y=1-\sqrt{\frac 12}$, $a=xy$, $b=x(1-y)$,
$c=(1-x)y$ and $d=0$, and (ii) $x=y=\sqrt{\frac 12}$, $a=xy$ and $b=c=d=0$, both
solutions yielding $\alpha=1+\sqrt 2$ when exactly half of the entries
in the bicluster $Y$ are 1's. This proves Theorem \ref{thm:main} for 0--1 valued
matrices and $L_1$-norm.

Notice that the above proof relies on the fact that the input 
matrix $X$ has only two types
of values. Therefore, the proof does not generalize to real valued
matrices.

An example of a matrix with approximation ratio of 2 is given by a
$4\times(4q-1)$ matrix
\begin{displaymath}
 X=
  \left(\begin{array}{ccc}
  0\ldots 0 & 1\ldots 1 & 0\ldots\ldots 0 \\
  0\ldots 0 & 1\ldots 1 & 1\ldots\ldots 1 \\
  1\ldots 1 & 0\ldots 0 & 0\ldots\ldots 0 \\
  1\ldots 1 & 0\ldots 0 & 1\ldots\ldots 1
  \end{array}\right)
\end{displaymath}
with $q$ columns in the first column group, $q$ columns in the second
column group and $2q-1$ columns in the third column group, clustered
to two row clusters, $K_r=2$, and one column cluster, $K_c=1$, at the
limit of large $q$. The optimal one-way clustering of rows is given by ${\cal
  R}=\{\{1,2\},\{3,4\}\}$, $L=8q-2$, and the optimal biclustering of rows by
${\cal R}^*=\{\{1,3\},\{2,4\}\}$, $L^*=4q$.

\subsection{$L_2$-norm and real valued matrix}
\label{sec:l2}

Consider now a real valued matrix $X$ and $L_2$-norm. We want to
prove Theorem \ref{thm:main} for the real valued biclusters $Y$ of $X$. 
To find the
approximation ratio, it suffices to show that
Equation (\ref{eq:thbound}) holds for each bicluster 
$Y\in{\mathbb{R}}^{n\times m}$, which are determined by $Y=X(R,C)$,
where $R\in{\cal R}$ and $C\in{\cal C}$.

Using the definitions of ${\cal V}(Y),$ $\vr(Y)$ and
$\vc (Y),$ we can write ${\cal V}(Y) = \vr (Y)+\vc (Y)- 
\sum_{i=1}^{n}{\sum_{j=1}^{m}{\left(Y(i,j)-\overline
      Y(i,j)\right)^2}}\le \vr (Y)+ \vc (Y)$, where $\overline Y(i,j)=
\mathrm{mean}(Y([n],j))+\mathrm{mean}(Y(i,[m]))-\mathrm{mean}(Y)$. Hence,
Equation (\ref{eq:thbound}) is satisfied for $L_{2}$-norm and
real valued matrices when $\alpha=2$.

\section{Conclusions}

We have shown that approximating the optimal biclustering with independent
row- and column-wise standard clusterings achieves a good approximation
guarantee. However in practice, standard one-way clustering
algorithms (such as $K$-means or $K$-median) are also approximate, 
and therefore, it is necessary to
multiply our ratio with the approximation ratio achieved by the
standard clustering algorithm (such as presented in
\cite{Arya01,Kanungo04})
to obtain the true approximation ratio of our scheme. Still, our
contribution shows that in
many practical applications of biclustering, it may be sufficient to
use a more straightforward standard clustering of rows and columns
instead of applying heuristic algorithms without performance
guarantees.


\section{Acknowledgments}

We thank Nikolaj Tatti for reading through the manuscript and giving useful
comments.

\end{document}